\documentclass[lettersize,journal]{IEEEtran}

\IEEEoverridecommandlockouts      

\usepackage[utf8]{inputenc}
\usepackage[T1]{fontenc}

\usepackage{xcolor}

\usepackage{graphics} 
\usepackage{epsfig} 

\usepackage{cite}

\usepackage{amsmath} 
\usepackage{comment}

\usepackage{amssymb}  
\usepackage{subcaption}
\usepackage{amsthm}
\usepackage{algpseudocode}
\usepackage[ruled, vlined, linesnumbered]{algorithm2e}
\newtheorem{remark}{Remark}
\newtheorem{theorem}{Theorem}
\newtheorem{lemma}{Lemma}

\newtheorem{definition}{Definition}
\title{\LARGE \bf
 Status Updating in Two-Way Delay Systems with Preemption
}

\author{
\IEEEauthorblockN{Jinxin Yang, Mohammad~Moltafet, and Hamid R. Sadjadpour}

\thanks{ J. Yang is with School of Electronic and Electrical Engineering, University of Leeds, UK, email: el23jy2@leeds.ac.uk.
 M. Moltafet and H. R. Sadjadpour are with Department of Electrical and Computer Engineering University of California Santa Cruz, email: \{mmoltafe, hamid\}@ucsc.edu.
 }
}

\begin{document}

\maketitle
\thispagestyle{empty}
\pagestyle{empty}

\begin{abstract}
We consider a status update system consisting of a sampler, a sink, and a controller located at the sink. The controller sends requests to the sampler to generate and transmit status updates. Packet transmissions from the controller to the sampler (reverse link) and from the sampler to the sink (forward link) experience random delays. 
The reverse and forward links are modeled as servers with geometric service times, referred to as the controller and sampler servers, respectively. Each server is equipped with a single buffer that stores an arriving packet when the server is busy. We adopt a \textit{preemption-in-waiting} policy on both links, whereby an arriving packet replaces the packet in the buffer whenever the buffer is full.
Our main goal is to determine the optimal generation times of request packets at the controller in order to minimize the long-term average age of information (AoI) at the sink. We formulate the problem as a Markov decision process (MDP) and derive the optimal stationary deterministic policy using the relative value iteration (RVI) algorithm. We prove the convergence of the algorithm. Numerical results show that the proposed system consistently outperforms baseline policies from prior work and reveal a threshold-based structure for the optimal policy.

\textit{Index Terms:} Two-way delay systems, optimal status updating, Markov decision process.
\end{abstract}

\section{Introduction}
Timely information delivery is a critical requirement in networked monitoring and control systems, where decisions rely on status updates transmitted over unreliable and delay-prone communication links. As Internet-of-Things (IoT) applications continue to grow, the Age of Information (AoI) has emerged as a widely used metric for quantifying information timeliness, measuring the time elapsed since the generation of the most recently received update \cite{Kaul2012}. 
Unlike conventional 
metrics, AoI captures the combined effects of sampling, queueing, and transmission dynamics on information freshness \cite{Kaul2012,Costa2017, SunCyrsPAWC2018,Yates2020,Yates2021,Moltafet2025}.

In many status update systems, updates are not generated periodically but are instead requested on demand by a controller due to system constraints such as limited energy. Such request-based architectures inherently introduce two-way delays, as a request must first reach the sampler before a corresponding status update can be generated and transmitted to the sink.
This bidirectional delay structure couples sampling decisions with network congestion, resulting in a fundamental tradeoff between update frequency and information freshness. 

In this paper, we consider a status update system consisting of a controller-sampler link (reverse link) and a sampler-sink link (forward link). Each link is modeled as a single server with geometric service times and equipped with a one-packet buffer. To enhance information freshness, 
we adopt a \textit{preemption-in-waiting} policy. In this setup, newly generated packets replace any packet currently residing in the buffer on both the forward and reverse links.  
The main goal of this work is to design a control policy that determines when to generate request packets to minimize the long-term average AoI at the sink. We formulate this problem as a Markov decision process (MDP) and solve for an optimal stationary deterministic policy using the relative value iteration (RVI) algorithm.  We prove the convergence of the algorithm and show the structure of the optimal policy in the numerical results section. 

Status updating under two-way delay has been studied from different perspectives in the literature, e.g., \cite{Tsai2021,Pan2022,Moltafet2023,Wang2024}.
The works in \cite{Tsai2021,Pan2022,Wang2024} consider systems where both status update transmissions and acknowledgments of the received updates experience random delays.
In \cite{Tsai2021}, the authors assume that a sensor can generate a new packet only after the previous packet has been delivered to the destination. They develop optimal status updating policies to minimize the average AoI or to perform remote estimation. Similarly, in \cite{Pan2022}, the authors assume that a fresh packet can be generated only after the previous packet has been delivered to the destination, and they derive an optimal sampling policy that minimizes a long-term average AoI penalty under a sampling rate constraint.
The work in \cite{Wang2024} studies AoI minimization in systems where both status update transmissions and acknowledgments experience random delays with arbitrary distributions. Based on the delayed acknowledgments, the sampler decides when to generate and transmit a fresh update. Compared to \cite{Tsai2021,Pan2022,Wang2024}, our system is controller-centric, where request messages experience delays, whereas these works consider sampler-centric systems in which acknowledgments experience delays; hence, the settings are fundamentally different.
In \cite{Moltafet2023}, the authors investigate optimal status update control in a system where a controller located at the sink sends request packets to the sampler. The sampler  generates a fresh update upon receiving a request. Both the request message transmission and the status update transmission experience random delays. The authors study the problem of minimizing the average AoI for systems with at most one and two active request packets, referred to as the 1-Packet system and the 2-Packet system, respectively.
Compared to \cite{Moltafet2023}, the present work removes the constraint on the maximum number of active requests and introduces fully flexible request generation through a preemption-in-waiting policy on both the forward and reverse links. 

\section{System Model}\label{sec:system_model}
We consider a status update system consisting of a sampler, a controller, and a sink, as shown in Fig.~\ref{fig:System}. Time is divided into unit slots indexed by $t \in \{0,1,\ldots\}$, where each slot corresponds to the interval $[t, t+1)$. The sampler observes an underlying stochastic process, while the sink aims to maintain fresh information about the process.

\subsubsection{Request Messaging}
We allow on-demand sampling: the \emph{controller} at the sink sends \emph{request packets} to the sampler, and upon reception, the sampler immediately takes a sample. Each request packet experiences a random delivery delay, which we model as a \emph{controller server} along the reverse link. The server has geometrically distributed service times with mean $1/\gamma$. We also consider a \emph{controller buffer} that stores arriving request packets while the server is busy.

\subsubsection{Status Updating}
Upon reception of a request packet at the end of slot $t-1$,
the sampler generates a sample at the beginning of slot $t$. The sample is
encapsulated in a status update packet that carries the measured value and a
generation timestamp. Each update packet experiences a random delay before reaching
the sink. We model this as a sampler server with geometric service time of mean $1/\mu$. We also consider a \textit{sampler buffer}
that stores an arriving packet while the server is busy.
\begin{figure}
    \centering
    \includegraphics[width=1\linewidth]{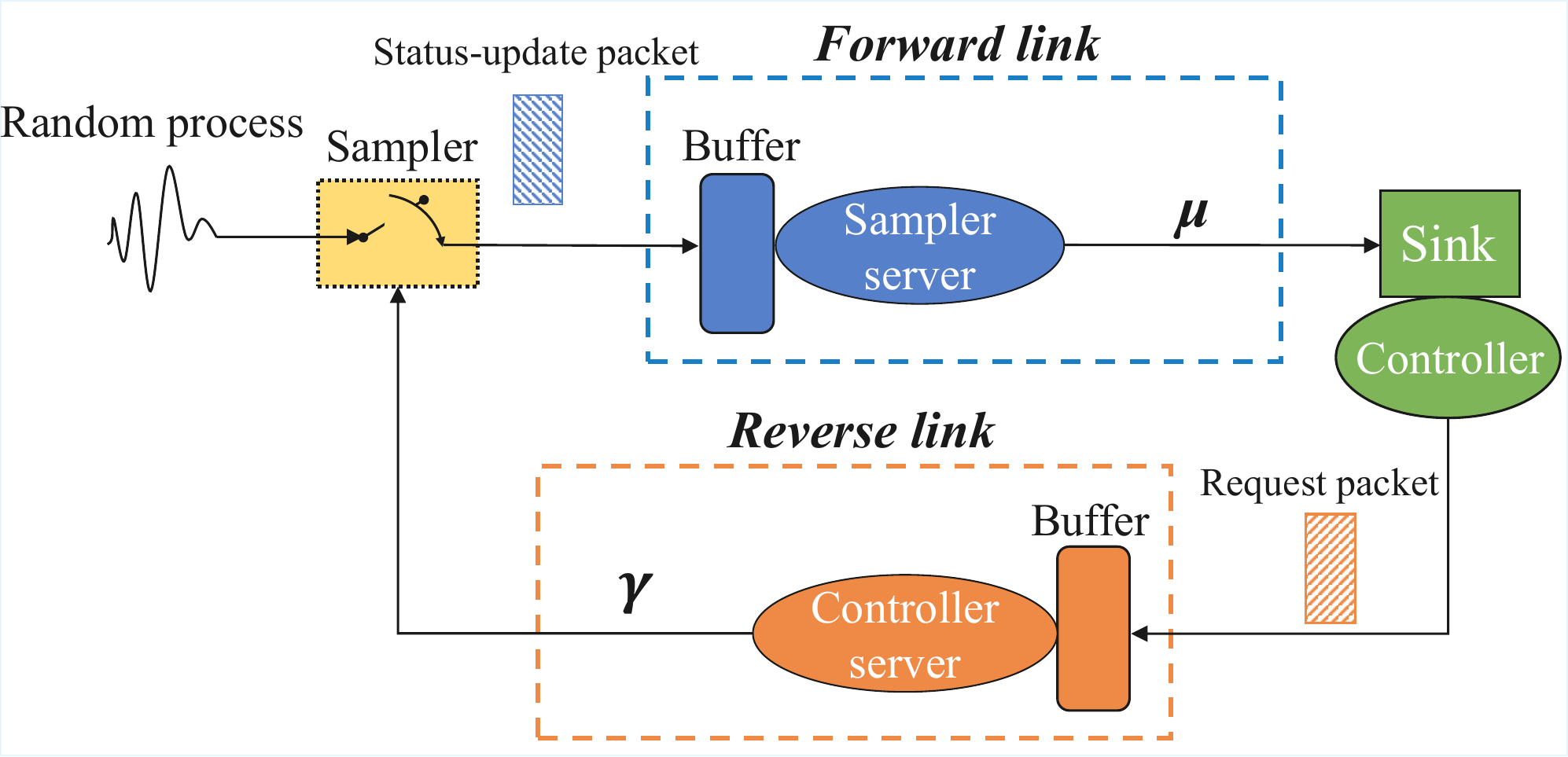}
    \caption{The considered system.}
    \label{fig:System}
\end{figure}

\subsubsection{Packet Management Policy} We adopt a preemption-in-waiting policy on both the reverse and forward links. According to this policy, 
when the server is idle, an arriving packet immediately enters service. 
If the server is busy, the packet is stored in the associated buffer. 
If the buffer is already occupied and a packet arrives, the newly arriving packet preempts the packet currently waiting in the buffer.
The packet in the buffer enters service immediately after the completion of the packet currently in service.

\subsubsection{AoI Definition}
The AoI at the sink is defined as the difference between the current time and the time stamp of the most recently received update. Let $t'_i$ denote the time slot at which the $i$-th status update is delivered to the sink, and let $N(t)$ denote the index of the most recently received status update packet by time slot $t$. Let $\delta^t$ denote the AoI at time slot $t$. Then, the AoI is given by the random process $\delta^t = t - t'_{N(t)}$.

\section{
Optimal Status Updating}\label{sec:OPTIMAL CONTROL FOR AVERAGE AoI MINIMIZATION}
This section develops an optimal control policy for the status update system. 
Our objective is to determine an optimal request-generation strategy  to minimize the  average AoI.

We make a standard assumption that the AoI is constrained by a sufficiently large value, denoted as $\bar{\Delta}$, i.e., $\delta^{t}\in\{1,2,\ldots,\bar{\Delta}\}$.
This assumption facilitates the design of optimal control policies by ensuring that the state space remains finite. Additionally, it reflects the idea that once the status becomes excessively outdated ,
any further increase in age becomes irrelevant for time-sensitive applications.

Next, we formulate the optimal status update control problem as an MDP and compute the optimal policy.
\subsubsection{State} 
The system state at time slot $t$ is denoted by  $s^{t} \in \mathcal{S}$, where $\mathcal{S}$ is the state space. The system state is defined as 
$s^{t} = \left( \delta^{t}, e_b^{t}, e_s^{t}, E_b^{t}, E_s^{t}, \Delta_b^{t}, \Delta_s^{t} \right)$
with the following seven elements:  
1) $\delta^{t}$ is the AoI at the sink,  
2) $e_b^{t} \in \{0,1\}$ is the occupancy indicator for the controller buffer: $e_b^{t}=1$ indicates that at slot $t$, there is a request packet at the controller buffer, and $e_b^{t}=0$ otherwise,
3) $e_s^{t} \in \{0,1\}$ is the occupancy indicator for the controller server,  
4) $E_b^{t} \in \{0,1\}$ is the occupancy indicator for the sampler buffer,  
5) $E_s^{t} \in \{0,1\}$ is the occupancy indicator for the sampler server,  
6) $\Delta_b^{t}$ is the age of the status update packet at the sampler buffer, and  
7) $\Delta_s^{t}$ is the age of the status update packet at the sampler server.  

\begin{remark}
The inclusion of the four occupancy variables $e_b^{t}$, $e_s^{t}$, $E_b^{t}$, and $E_s^{t}$ in the state description implicitly assumes that the controller has complete visibility of the entire system. In practical implementations, such full observability may not be available. Therefore, the optimal policy derived under this assumption serves as a benchmark, providing a performance upper bound for policies that operate with 
partial 
information.
\end{remark}

\subsubsection{Action}
At each time slot $t$, the controller determines an action $a^t \in \mathcal{A}=\{0,1\}$;  $a^t=1$ corresponds to generating a request packet at the beginning of slot $t$, 
and  $a^t=0$ otherwise.

\subsubsection{Policy} A policy $\pi$ is 
a function that maps each state in $\mathcal{S}$ to an action in $\mathcal{A}$. In particular, for a given system state $s^t$, the policy specifies the action $a^t$ to be taken at time slot $t$.

\subsubsection{Problem Formulation} For a given initial state $s^0 \in \mathcal{S}$, the long-term time-average AoI under a policy $\pi$ is defined as
\begin{equation}
J_{\pi}(s^0) = \limsup_{T \to \infty} \frac{1}{T} \sum_{t=0}^{T-1} \mathbb{E}_{\pi} \left[ \delta^t \mid s^0 \right],
\end{equation}
where the expectation is taken over the randomness in the actions selected according to policy 
$\pi$ and the stochastic service processes at both the controller and sampler servers.

Our goal is to determine an optimal policy $\pi^{*}$ that minimizes the 
long-term average AoI
\begin{equation}\label{eq:optimal policy}
    \pi^*(s^0) = \arg \min_{\pi \in \Pi} J_\pi(s^0).
\end{equation}
Next, we cast the problem as an MDP and solve it
via RVI.

\subsubsection{MDP Modeling of Problem} 

The problem is modeled as an MDP characterized by the tuple 
$(\mathcal{S}, \mathcal{A}, \Pr(s^{t+1}\mid s^t, a^t), C(s^t, a^t))$, 
where $C(s^t, a^t)$ represents the immediate cost incurred when action $a^t$ is taken in state $s^t$. 
In our formulation, the one-step cost is defined as $C(s^t, a^t) = \delta^{t+1}$, namely the AoI observed at the beginning of the subsequent slot $t+1$.
The transition probability
$\Pr(s' \mid s, a) = \Pr(s^{t+1}\mid s^t, a^t)$ 
specifies the probability that the system transfers from the state $s = s^t$ to the next state $s' = s^{t+1}$ when action $a = a^t$ is taken.

For notational convenience in expressing $\Pr(s' \mid s,a)$, we introduce the shorthand $\bar{\gamma} \triangleq 1-\gamma$ and $\bar{\mu} \triangleq 1-\mu$, along with the operator 
\begin{align}
    [\beta]^{\bar{\Delta}} \triangleq \min\{\beta+1,\bar{\Delta}\}.
\end{align}

We further use the symbol $\star$ to denote an undefined age value in cases where the sampler server is idle or the sampler buffer contains no packet.

The  transition probability from the current state 
$s = (\delta, e_b, e_s, E_b, E_s, \Delta_b, \Delta_s)$ 
to the subsequent state 
$s' = (\delta', e_b', e_s', E_b', E_s', \Delta_b', \Delta_s')$ under action $a$, shown as $\Pr\left(s' \mid s=(\delta, 0,0,0,0, \star, \star), a\right)$, is characterized in the following. 

\allowdisplaybreaks[4]

\begin{subequations} \label{eq:state-transition-3}
\begin{align}
&\Pr\!\left(s' \mid s=(\delta, 0,0,0,0, \star, \star), a\right) \notag\\
&\quad=\begin{cases}
1, & a=0,\, s'=\left([\delta]^{\bar{\Delta}}, 0,0,0,0, \star, \star\right), \\
\gamma, & a=1,\, s'=\left([\delta]^{\bar{\Delta}}, 0,0,0,1, \star, 0\right), \\
\bar{\gamma}, & a=1,\, s'=\left([\delta]^{\bar{\Delta}}, 0,1,0,0, \star, \star\right), \\
0, & \text{otherwise}.
\end{cases} \\[8pt]
&\Pr\!\left(s' \mid s=(\delta, 1,1,0,0, \star, \star), a\right) \notag\\
&\quad=\begin{cases}
\gamma, & a=0,\, s'=\left([\delta]^{\bar{\Delta}}, 0,1,0,1, \star, 0\right), \\
\bar{\gamma}, & a=0,\, s'=\left([\delta]^{\bar{\Delta}}, 1,1,0,0, \star, \star\right), \\
\gamma, & a=1,\, s'=\left([\delta]^{\bar{\Delta}}, 0,1,0,1, \star, 0\right), \\
\bar{\gamma}, & a=1,\, s'=\left([\delta]^{\bar{\Delta}}, 1,1,0,0, \star, \star\right), \\
0, & \text{otherwise}.
\end{cases} \\[8pt]
&\Pr\!\left(s' \mid s=(\delta, 0,1,0,0, \star, \star), a\right) \notag\\
&\quad=\begin{cases}
\gamma, & a=0,\, s'=\left([\delta]^{\bar{\Delta}}, 0,0,0,1, \star, 0\right), \\
\bar{\gamma}, & a=0,\, s'=\left([\delta]^{\bar{\Delta}}, 0,1,0,0, \star, \star\right), \\
\gamma, & a=1,\, s'=\left([\delta]^{\bar{\Delta}}, 0,1,0,1, \star, 0\right), \\
\bar{\gamma}, & a=1,\, s'=\left([\delta]^{\bar{\Delta}}, 1,1,0,0, \star, \star\right), \\
0, & \text{otherwise}.
\end{cases} \\[8pt]
&\Pr\!\left(s' \mid s=\left(\delta, 0,1,0,1, \star, \Delta_{\mathrm{s}}\right), a\right) \notag\\
&\quad=\begin{cases}
\bar{\gamma} \bar{\mu}, & a=0,\, s'=\left([\delta]^{\bar{\Delta}}, 0,1,0,1, \star,\left[\Delta_{\mathrm{s}}\right]^{\bar{\Delta}}\right), \\
\gamma \bar{\mu}, & a=0,\, s'=\left([\delta]^{\bar{\Delta}}, 0,0,1,1,0,\left[\Delta_{\mathrm{s}}\right]^{\bar{\Delta}}\right), \\
\bar{\gamma} \mu, & a=0,\, s'=\left(\left[\Delta_{\mathrm{s}}\right]^{\bar{\Delta}}, 0,1,0,0, \star, \star\right), \\
\gamma \mu, & a=0,\, s'=\left(\left[\Delta_{\mathrm{s}}\right]^{\bar{\Delta}}, 0,0,0,1, \star, 0\right), \\
\bar{\gamma} \bar{\mu}, & a=1,\, s'=\left([\delta]^{\bar{\Delta}}, 1,1,0,1, \star,\left[\Delta_{\mathrm{s}}\right]^{\bar{\Delta}}\right), \\
\gamma \bar{\mu}, & a=1,\, s'=\left([\delta]^{\bar{\Delta}}, 0,1,1,1,0,\left[\Delta_{\mathrm{s}}\right]^{\bar{\Delta}}\right), \\
\bar{\gamma} \mu, & a=1,\, s'=\left(\left[\Delta_{\mathrm{s}}\right]^{\bar{\Delta}}, 1,1,0,0, \star, \star\right), \\
\gamma \mu, & a=1,\, s'=\left(\left[\Delta_{\mathrm{s}}\right]^{\bar{\Delta}}, 0,1,0,1, \star, 0\right), \\
0, & \text{otherwise}.
\end{cases}
\end{align}
\end{subequations}

\addtocounter{equation}{-1}
\begin{subequations}\label{eq:state-transition-3-f-i}
\setcounter{equation}{4}
\begin{align}
&\Pr\!\left(s' \mid s=\left(\delta, 0,0,1,1, \Delta_{\mathrm{b}}, \Delta_{\mathrm{s}}\right), a\right) \notag\\
&\quad=\begin{cases}
\mu, & a=0,\, s'=\left(\left[\Delta_{\mathrm{s}}\right]^{\bar{\Delta}}, 0,0,0,1, \star, \left[\Delta_{\mathrm{b}}\right]^{\bar{\Delta}}\right), \\
\bar{\mu}, & a=0,\, s'=\left([\delta]^{\bar{\Delta}}, 0,0,1,1, \left[\Delta_{\mathrm{b}}\right]^{\bar{\Delta}}, \left[\Delta_{\mathrm{s}}\right]^{\bar{\Delta}}\right), \\
\bar{\gamma} \bar{\mu}, & a=1,\, s'=\left([\delta]^{\bar{\Delta}}, 0,1,1,1, \left[\Delta_{\mathrm{b}}\right]^{\bar{\Delta}}, \left[\Delta_{\mathrm{s}}\right]^{\bar{\Delta}}\right), \\
\gamma \bar{\mu}, & a=1,\, s'=\left([\delta]^{\bar{\Delta}}, 0,0,1,1, 0, \left[\Delta_{\mathrm{s}}\right]^{\bar{\Delta}}\right), \\
\bar{\gamma} \mu, & a=1,\, s'=\left(\left[\Delta_{\mathrm{s}}\right]^{\bar{\Delta}}, 0,1,0,1, \star, \left[\Delta_{\mathrm{b}}\right]^{\bar{\Delta}}\right), \\
\gamma \mu, & a=1,\ s'=\left(\left[\Delta_{\mathrm{s}}\right]^{\bar{\Delta}}, 0,0,0,1, *,0\right), \\
0, & \text{otherwise}.
\end{cases} \\[6pt]
&\Pr\!\left(s' \mid s=\left(\delta, 0,0,0,1, \star, \Delta_{\mathrm{s}}\right),\, a\right) \notag\\
&\quad=\begin{cases}
\mu, & a=0,\, s'=\left(\left[\Delta_{\mathrm{s}}\right]^{\bar{\Delta}}, 0,0,0,0, \star,\star\right), \\
\bar{\mu}, & a=0,\, s'=\left([\delta]^{\bar{\Delta}}, 0,0,0,1,\star,\left[\Delta_{\mathrm{s}}\right]^{\bar{\Delta}}\right), \\
\bar{\gamma} \bar{\mu}, & a=1,\, s'=\left([\delta]^{\bar{\Delta}}, 0,1,0,1, \star,\left[\Delta_{\mathrm{s}}\right]^{\bar{\Delta}}\right), \\
\gamma \bar{\mu}, & a=1,\, s'=\left([\delta]^{\bar{\Delta}}, 0,0,1,1,0,\left[\Delta_{\mathrm{s}}\right]^{\bar{\Delta}}\right), \\
\bar{\gamma} \mu, & a=1,\, s'=\left(\left[\Delta_{\mathrm{s}}\right]^{\bar{\Delta}}, 0,1,0,0, \star, \star\right), \\
\gamma \mu, & a=1,\, s'=\left(\left[\Delta_{\mathrm{s}}\right]^{\bar{\Delta}}, 0,0,0,1,\star, 0\right), \\
0, & \text{otherwise}.
\end{cases} \\[6pt]
&\Pr\!\left(s' \mid s=\left(\delta, 1,1,1,1, \Delta_{\mathrm{b}}, \Delta_{\mathrm{s}}\right), a\right) \notag\\
&\quad=\begin{cases}
\bar{\gamma} \bar{\mu}, & a=0,\, s'=\left([\delta]^{\bar{\Delta}}, 1,1,1,1, \left[\Delta_{\mathrm{b}}\right]^{\bar{\Delta}}, \left[\Delta_{\mathrm{s}}\right]^{\bar{\Delta}}\right), \\
\gamma \bar{\mu}, & a=0,\, s'=\left([\delta]^{\bar{\Delta}}, 0,1,1,1, 0, \left[\Delta_{\mathrm{s}}\right]^{\bar{\Delta}}\right), \\
\bar{\gamma} \mu, & a=0,\, s'=\left(\left[\Delta_{\mathrm{s}}\right]^{\bar{\Delta}}, 1,1,0,1, \star, \left[\Delta_{\mathrm{b}}\right]^{\bar{\Delta}}\right), \\
\gamma \mu, & a=0,\, s'=\left(\left[\Delta_{\mathrm{s}}\right]^{\bar{\Delta}}, 0,1,0,1, *,0\right), 
\\
\bar{\gamma} \bar{\mu}, & a=1,\, s'=\left([\delta]^{\bar{\Delta}}, 1,1,1,1, \left[\Delta_{\mathrm{b}}\right]^{\bar{\Delta}}, \left[\Delta_{\mathrm{s}}\right]^{\bar{\Delta}}\right), \\
\gamma \bar{\mu}, & a=1,\, s'=\left([\delta]^{\bar{\Delta}}, 0,1,1,1, 0, \left[\Delta_{\mathrm{s}}\right]^{\bar{\Delta}}\right), \\
\bar{\gamma} \mu, & a=1,\, s'=\left(\left[\Delta_{\mathrm{s}}\right]^{\bar{\Delta}}, 1,1,0,1, \star, \left[\Delta_{\mathrm{b}}\right]^{\bar{\Delta}}\right), \\
\gamma \mu, & a=1,\, s'=\left(\left[\Delta_{\mathrm{s}}\right]^{\bar{\Delta}}, 0,1,0,1,*,0\right),\\
0, & \text{otherwise}.
\end{cases} \\[6pt]
&\Pr\!\left(s' \mid s=\left(\delta, 1,1,0,1,\star, \Delta_{\mathrm{s}}\right), a\right) \notag\\
&\quad=\begin{cases}
\bar{\gamma}\bar{\mu}, & a=0,\, s'=\left([\delta]^{\bar{\Delta}}, 1,1,0,1, \star, \left[\Delta_{\mathrm{s}}\right]^{\bar{\Delta}}\right), \\
\gamma\bar{\mu}, & a=0,\, s'=\left([\delta]^{\bar{\Delta}}, 0,1,1,1, 0, \left[\Delta_{\mathrm{s}}\right]^{\bar{\Delta}}\right), \\
\bar{\gamma}\mu, & a=0,\, s'=\left(\left[\Delta_{\mathrm{s}}\right]^{\bar{\Delta}}, 1,1,0,0, \star, \star\right), \\
\gamma\mu, & a=0,\, s'=\left(\left[\Delta_{\mathrm{s}}\right]^{\bar{\Delta}}, 0,1,0,1, \star, 0\right), \\
\bar{\gamma}\bar{\mu}, & a=1,\, s'=\left([\delta]^{\bar{\Delta}}, 1,1,0,1, \star, \left[\Delta_{\mathrm{s}}\right]^{\bar{\Delta}}\right), \\
\gamma\bar{\mu}, & a=1,\, s'=\left([\delta]^{\bar{\Delta}}, 0,1,1,1, 0, \left[\Delta_{\mathrm{s}}\right]^{\bar{\Delta}}\right), \\
\bar{\gamma}\mu, & a=1,\, s'=\left(\left[\Delta_{\mathrm{s}}\right]^{\bar{\Delta}}, 1,1,0,0, \star, \star\right), \\
\gamma\mu, & a=1,\, s'=\left(\left[\Delta_{\mathrm{s}}\right]^{\bar{\Delta}}, 0,1,0,1, \star, 0\right), \\
0, & \text{otherwise}.
\end{cases} \\[6pt]
&\Pr\!\left(s' \mid s=\left(\delta, 0,1,1,1, \Delta_{\mathrm{b}}, \Delta_{\mathrm{s}}\right), a\right) \notag\\
&\quad=\begin{cases}
\bar{\gamma}\bar{\mu}, & a=0,\, s'=\left([\delta]^{\bar{\Delta}}, 0,1,1,1, \left[\Delta_{\mathrm{b}}\right]^{\bar{\Delta}}, \left[\Delta_{\mathrm{s}}\right]^{\bar{\Delta}}\right), \\
\gamma\bar{\mu}, & a=0,\, s'=\left([\delta]^{\bar{\Delta}}, 0,0,1,1, 0, \left[\Delta_{\mathrm{s}}\right]^{\bar{\Delta}}\right), \\
\bar{\gamma}\mu, & a=0,\, s'=\left(\left[\Delta_{\mathrm{s}}\right]^{\bar{\Delta}}, 0,1,0,1, \star, \left[\Delta_{\mathrm{b}}\right]^{\bar{\Delta}}\right), \\
\gamma\mu, & a=0,\, s'=\left(\left[\Delta_{\mathrm{s}}\right]^{\bar{\Delta}}, 0,0,0,1, *, 0\right), \\
\bar{\gamma}\bar{\mu}, & a=1,\, s'=\left([\delta]^{\bar{\Delta}}, 1,1,1,1, \left[\Delta_{\mathrm{b}}\right]^{\bar{\Delta}}, \left[\Delta_{\mathrm{s}}\right]^{\bar{\Delta}}\right), \\
\gamma\bar{\mu}, & a=1,\, s'=\left([\delta]^{\bar{\Delta}}, 0,1,1,1, 0, \left[\Delta_{\mathrm{s}}\right]^{\bar{\Delta}}\right), \\
\bar{\gamma}\mu, & a=1,\, s'=\left(\left[\Delta_{\mathrm{s}}\right]^{\bar{\Delta}}, 1,1,0,1, \star, \left[\Delta_{\mathrm{b}}\right]^{\bar{\Delta}}\right), \\
\gamma\mu, & a=1,\, s'=\left(\left[\Delta_{\mathrm{s}}\right]^{\bar{\Delta}}, 0,1,0,1, *, 0 \right), 
\\
0, & \text{otherwise}.
\end{cases}
\end{align}
\end{subequations}

Based on the MDP construction above, the problem is reformulated as the following MDP problem
\begin{equation} \label{eq:markov decesion}
\pi^*(s^0) = \underset{\pi \in \Pi}{\arg\min}
\left\{
  \limsup_{T \to \infty} \frac{1}{T} \sum_{t=0}^{T-1}
  \mathbb{E}\left[ C(s^t, a^t) \mid s^0 \right]
\right\}.
\end{equation}
\subsubsection{Optimal Policy} According to \cite[Proposition 4.2.1]{Bertsekas2007}, for a given initial state $s^0$, if there exist a scalar $\bar{V}$ and a set of values $\{V^*_s : s \in \mathcal{S}\}$ that satisfy the following equation:
\begin{equation}\label{eq:bellman}
\bar{V} + V^*_s =  \min_{a \in \mathcal{A}} \left\{
    \sum_{s' \in \mathcal{S}} \Pr(s'|s, a) \left(C(s, a) + V^*_{s'}\right)
\right\},
\end{equation}
then $\bar{V}$ represents the optimal average cost of the original problem,  $J_{\pi^*}(s^0)$, and  actions that achieves the minimum in~\eqref{eq:bellman} gives rise to an optimal stationary policy,  $\pi^*(s^0)$.

We next introduce a condition that guarantees
the existence of a solution to the Bellman equation
in~\eqref{eq:bellman}, and ensures that the resulting optimal
average cost does not depend on the initial state $s^{0}$\cite[Proposition 4.2.3]{Bertsekas2007}.
\begin{definition}[Weak Accessibility]
An MDP satisfies the weak accessibility condition if its state space can be partitioned into $\mathcal{S}_c$ and $\mathcal{S}_t$ such that:
I) For any $s, s' \in \mathcal{S}_c$, there exists a stationary policy under which the probability of reaching $s'$ from $s$ in some finite number of steps is non-zero, and
II) All states in $\mathcal{S}_t$ are transient under every stationary policy\cite[Definition 4.2.2]{Bertsekas2007}
\end{definition}

We now verify that the MDP in~\eqref{eq:markov decesion} satisfies the weak accessibility condition.

\begin{lemma}
For any $(\mu, \gamma) \in \{ (\theta_1, \theta_2) \mid 0 < \theta_1, \theta_2 \le 1,\ (\theta_1,\theta_2) \ne (1,1) \}$,
the weak accessibility condition holds for the MDP problem in~\eqref{eq:markov decesion}.
\end{lemma}
\begin{proof}
The weak accessibility condition holds if every deterministic policy in $\Pi$ induces a \emph{unichain} \cite[Proposition 4.2.5]{Bertsekas2007}.  
A deterministic policy is called unichain if, under the Markov chain it defines, there exists a state that can be reached from every other state with nonzero probability \cite[Exercise 4.3]{Gallager2013}.
In the MDP formulation~\eqref{eq:markov decesion}, we verify this property for all deterministic policies as follows:
\begin{itemize}
    \item[I)] If $\gamma \ne 1$ and $0 < \mu \le 1$, the state $(\bar{\Delta}, 1, 1, 0, 0, \star, \star)$ is reachable from any other state.
    \item[II)] If $\gamma = 1$ and $0 < \mu < 1$, the state $(\bar{\Delta}, 0, 0, 1, 1, 0, \bar{\Delta})$ is reachable from any other state.
\end{itemize}

For case~I, this follows because when both the controller server and buffer contain packets, the probability that all transmissions fail over $\bar{\Delta}$ consecutive slots is strictly positive.  
A similar argument applies to case~II.  
Therefore, every deterministic policy induces a unichain, and the weak accessibility condition holds.
\end{proof}
\begin{algorithm}[tbp] 
\SetAlgoLined
\SetKwInOut{KwInit}{Initialization}
\caption{RVI Algorithm}\label{alg:RVI}
\KwIn{1) State transition probabilities $\Pr(s' \mid s, a)$, 
2) stopping criterion threshold $\epsilon$.}

\KwInit{1) Set $V_s^\ast = 0,\ \forall s \in \mathcal{S}$, 
2) select an arbitrary reference state $s_{\mathrm{ref}} \in \mathcal{S}$, 
and 3) set $\phi > \epsilon$.}

\While{$\phi > \epsilon$}{
    \For{$s \in \mathcal{S}$}{
        \mbox{\!$\pi_s^\ast \!\!\gets\! 
        \arg\min_{a \in \mathcal{A}}\!
        \Big\{\!\!\textstyle\sum_{s' \in \mathcal{S}}
        \Pr(s' \!\mid\! s, a)
        \big(c(s,a) \!+\! V_{s'}^\ast\big)
        \!\!\Big\}$\!}\;
        
        \!\!$V_s \gets 
        \sum_{s' \in \mathcal{S}}
        \Pr(s' \mid s, \pi_s^\ast)
        \big(c(s,\pi_s^\ast) + V_{s'}^\ast\big)
        - V_{s_{\mathrm{ref}}}^\ast$\;
    }
    
    $\phi \gets \max_{s \in \mathcal{S}} \lvert V_s - V_s^\ast \rvert$\;
    $V_s^\ast \gets V_s,\ \forall s \in \mathcal{S}$\;
}

\KwOut{1) Optimal policy $\pi^\ast$, 2) optimal objective value 
$J_{\pi^\ast} = V_{s_{\mathrm{ref}}}^\ast$.}

\end{algorithm}
Under the weak accessibility condition, the optimal
stationary deterministic policy $\pi^{*}$ and the associated
average cost $J^{*}$ are well defined and independent of
$s^{0}$ \cite[Proposition 4.2.6]{Bertsekas2007}. 

The Bellman equation~\eqref{eq:bellman} can be solved using the RVI algorithm as described in \cite[Section 4.3]{Bertsekas2007}. The steps of the RVI algorithm are outlined in Algorithm~\ref{alg:RVI}. In each iteration, for every state $s \in \mathcal{S}$, the optimal action $\pi_s^{*}$
and value function $V^*_s$ are obtained
according to Lines~3 and~4, where an arbitrary reference
state $s_{\mathrm{ref}} \in \mathcal{S}$ is selected to normalize
the value function. Upon convergence, the algorithm returns
the optimal policy $\pi^{*}$ and the minimum average AoI,
given by $J_{\pi^{*}} = V^{*}_{s_{\mathrm{ref}}}$. 
%
%
The following theorem establishes that the RVI algorithm in  Algorithm~\ref{alg:RVI} converges to an optimal deterministic policy.

\begin{theorem}
The RVI algorithm in Algorithm 1 converges to an optimal deterministic policy and yields the optimal value of the average AoI. 
\end{theorem}

\begin{proof}
If the Markov chain induced by every deterministic policy is both unichain and aperiodic, the RVI algorithm is guaranteed to converge to an optimal policy and to return the corresponding optimal value of the average AoI \cite
[Page 209]{Bertsekas2007}.  
In the proof of Lemma~1, we demonstrated that each deterministic policy induces a unichain Markov chain.  
Moreover, since the recurrent states $(\bar{\Delta}, 1, 1, 0, 0, \star, \star)$ (for $\gamma \ne 1$ and $0 < \mu \le 1$) and $(\bar{\Delta}, 0, 0, 1, 1, 0, \bar{\Delta})$ (for $\gamma = 1$ and $0 < \mu < 1$) possess self-transitions, the Markov chain induced by any deterministic policy is aperiodic \cite[Exercise 4.1]{Gallager1996}.
\end{proof}
\section{Numerical Results}
In this section, we present numerical results to evaluate the performance of the proposed policy. We compare the performance of the system with the 1-Packet system and 2-Packet system studied in~\cite{Moltafet2023}.
The optimal policies are obtained by solving the MDP formulated in Section~\ref{sec:OPTIMAL CONTROL FOR AVERAGE AoI MINIMIZATION} using the RVI algorithm. Unless otherwise specified, Algorithm~\ref{alg:RVI} is executed with a maximum AoI value $\bar{\Delta} = 55$ and a convergence threshold $\epsilon = 5\times10^{-4}$.

\subsection{Maximum  Value for the AoI}
To ensure that the selected maximum value for the AoI $\bar{\Delta}$ does not affect the system performance, we first determine an appropriate AoI upper bound, denoted by $\bar{\Delta}$. This is achieved by gradually increasing $\bar{\Delta}$ until the average AoI under the optimal policy converges and remains unchanged. Fig.~\ref{fig:saturation point} illustrates the average AoI of the optimal policy obtained via the RVI algorithm as a function of $\bar{\Delta}$ for different values of the service rate parameters $\mu$ and $\gamma$.

As observed, when $\bar{\Delta}$ is chosen to be too small, the resulting average AoI decreases significantly, which may lead to misleading conclusions regarding the system performance. Therefore, selecting a sufficiently large upper bound for AoI is crucial for ensuring accurate numerical results. Specifically, for the service rates $\mu = 0.2$ and $\gamma = 0.4$, an AoI upper bound of $\bar{\Delta} > 50$ provides sufficient accuracy for practical purposes. Similarly, for $\mu = 0.4$ and $\gamma = 0.4$, the saturation point is observed at approximately $\bar{\Delta} > 25$.

\begin{figure}
    \centering
    \includegraphics[width=1\linewidth]{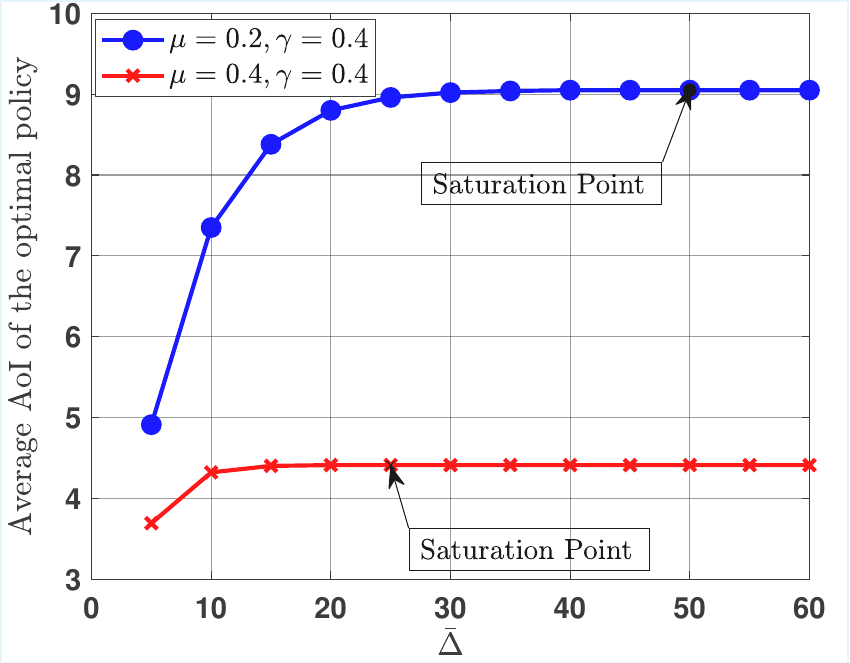}
    \caption{Average AoI of the optimal policy as a function of $\bar{\Delta}$ under different values of $\mu$ and $\gamma$.}
    \label{fig:saturation point}
\end{figure}

\subsection{Comparison with 1-Packet System and 2-Packet System}

Fig.~\ref{fig:y07} shows the average AoI as a function of the service rate $\mu$ with $\gamma=0.7$, under the optimal policies for both the 1-Packet and 2-Packet systems from~\cite{Moltafet2023}, as well as the proposed policy in Section~\ref{sec:system_model}.

As it can be seen, the proposed policy consistently achieves a lower average AoI compared to both the 1-Packet and 2-Packet systems. Additionally, for all systems, the average AoI decreases monotonically as the service rate $\mu$ increases, reflecting the intuition that faster service leads to fresher status updates and, consequently, lower AoI.

\begin{figure}
    \centering
    \includegraphics[width=1\linewidth]{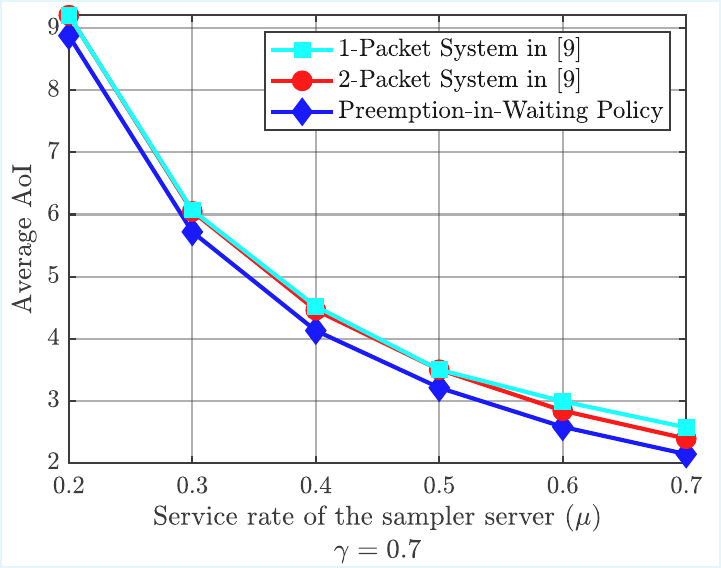}
    \caption{Average AoI versus service rate $\mu$ with $\gamma = 0.7$.}
    \label{fig:y07}
\end{figure}

\subsection{Structure of the Optimal Policy}

Fig.~\ref{fig:0001} illustrates the structure of the optimal policy when only the sampler server is occupied, i.e., the state $s = (\delta, 0, 0, 0, 1, \star, \Delta_{\mathrm{s}})$.\footnote{Similar behavior is observed for other states; however, due to space limitations, we only present the structure of the optimal policy for the state $s = (\delta, 0, 0, 0, 1, \star, \Delta_{\mathrm{s}})$.} As shown in Fig.~\ref{fig:0001}, the optimal policy exhibits a threshold behavior with respect to the age of the packet currently being served at the sampler server, $\Delta_{\mathrm{s}}$.
Specifically, for the case $\gamma = 0.1$ and $\mu = 0.1$, the optimal action is to remain idle when $\Delta_{\mathrm{s}} \leq 5$. Once $\Delta_{\mathrm{s}}$ exceeds this threshold, the optimal action switches to sending a new request packet. 



Moreover, it can be observed that for a fixed value of $\mu$, the threshold with respect to $\Delta_{\mathrm{s}}$ first increases and then decreases as $\gamma$ increases. This non-monotonic behavior can be explained by the interplay of two competing effects.
On one hand, as $\gamma$ increases from a small value, the controller server becomes faster, increasing the likelihood that the packet currently under service will be delivered within a relatively shorter time. When $\mu$ is fixed and not sufficiently large, prematurely sending new requests may lead to delivering to consecutive packets to the sink with age values close to each other and there will be a waste in the sampler server resources.  In this case, waiting for the completion of the ongoing status update becomes more beneficial, which results in an increase in the threshold.

On the other hand, when $\gamma$ becomes sufficiently large, the controller server can generate and deliver status updates at a much higher rate. As a result, the age of the packet in the sampler buffer remains relatively small, and the system can quickly refresh outdated information. In this case, proactively sending a new request becomes more advantageous, leading to a decrease in the threshold.
\begin{figure}
\centering
\setlength{\tabcolsep}{2pt}
\begin{subfigure}{0.45\textwidth}
    \centering
    \includegraphics[width=1\linewidth]{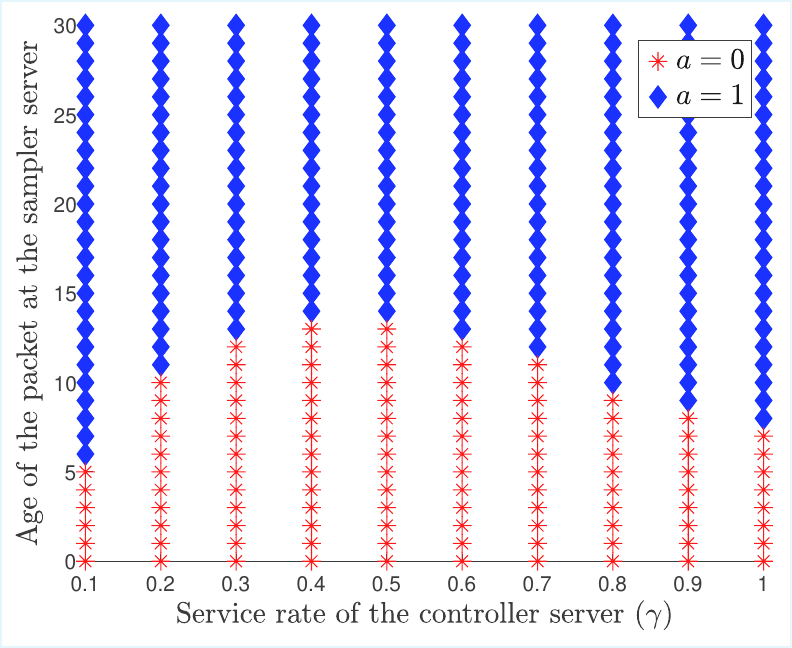}
    \caption{$\mu = 0.1$}
\end{subfigure}
\hspace{0.02\textwidth}
\begin{subfigure}{0.45\textwidth}
    \centering
    \includegraphics[width=1\linewidth]{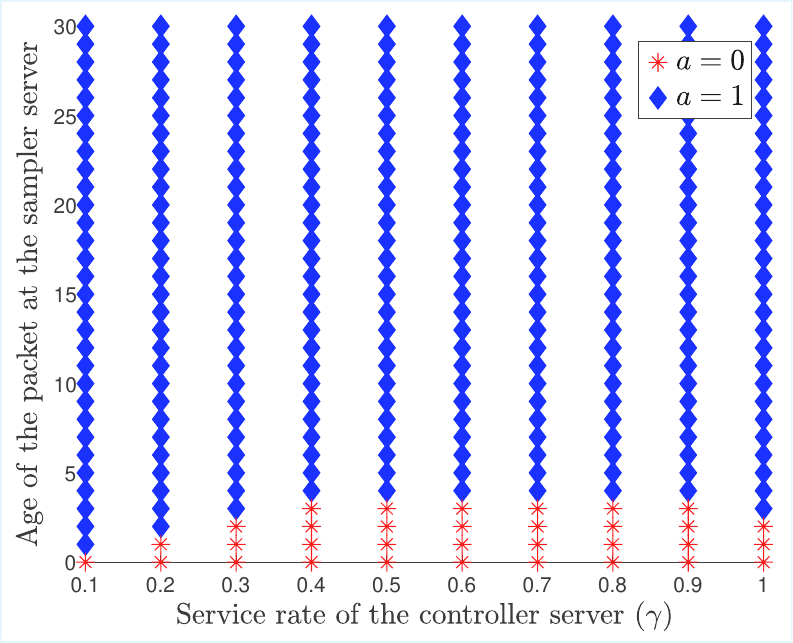}
    \caption{$\mu = 0.2$}
\end{subfigure}


\caption{Structure of the optimal policy for the system in state
$s = (\delta,0,0,0,1,\star,\Delta_{\mathrm{s}})$.}
\label{fig:0001}
\end{figure}

\section{Conclusion}
We studied the optimal status updating in a two-way system with preemption, where a controller dynamically decides when to generate request packets to minimize the average AoI. We formulated the problem as an MDP,  derived the optimal deterministic policy using the RVI algorithm, and proved the convergence of the algorithm.
Through numerical evaluations, we demonstrated that the proposed system consistently outperforms the systems without preemption in terms of average AoI.
Moreover, our results revealed a non-monotonic  threshold-based structure of the optimal policy.


\bibliographystyle{IEEEtran}  
\bibliography{reference}

\end{document}